\documentclass[aps, prl, twocolumn, superscriptaddress, nofootinbib, showpacs]{revtex4-1}
\usepackage{amsfonts,amssymb,amsmath}
\usepackage{amsthm}
\usepackage{graphics,graphicx,epsfig}
\usepackage[dvipsnames]{xcolor}
\usepackage[normalem]{ulem}
\usepackage{algorithm}
\usepackage{algorithmicx}
\usepackage{algpseudocode}
\usepackage{multirow}

\newcommand{\ket}[1]{\mbox{$ | #1 \rangle $}}
\newcommand{\bra}[1]{\mbox{$ \langle #1 | $}}
\newcommand{\braket}[2]{\mbox{$ \langle #1 | #2 \rangle $}}
\newcommand{\rhoML}{\widehat{\rho}_{\textsc{m}}}
\newcommand{\tr}{\mathrm{tr}}
\newcommand{\cF}{{\cal F}}
\newcommand{\cL}{{\cal L}}
\newcommand{\cC}{{\cal C}}
\newcommand{\cQ}{{\cal Q}}
\newcommand{\cS}{{\mathcal S}}
\newcommand{\cG}{{\mathcal G}}

\newtheorem{theorem}{Theorem}
\newtheorem{statement}{Statement}
\newtheorem{proposition}{Proposition}

\begin{document}
\title{Convex Optimization over Classes of Multiparticle Entanglement}
\author{Jiangwei Shang}
\email{jiangwei.shang@bit.edu.cn}
\affiliation{Naturwissenschaftlich-Technische Fakult{\"a}t, Universit{\"a}t Siegen, Walter-Flex-Stra{\ss}e 3, 57068 Siegen, Germany}
\affiliation{Beijing Key Laboratory of Nanophotonics and Ultrafine Optoelectronic Systems, School of Physics, Beijing Institute of Technology, Beijing 100081, China}
\author{Otfried G{\"u}hne}
\email{otfried.guehne@uni-siegen.de}
\affiliation{Naturwissenschaftlich-Technische Fakult{\"a}t, Universit{\"a}t Siegen, Walter-Flex-Stra{\ss}e 3, 57068 Siegen, Germany}

\date{\today}

\begin{abstract}
A well-known strategy to characterize multiparticle entanglement utilizes
the notion of stochastic local operations and classical communication (SLOCC),
but characterizing the resulting entanglement classes is difficult. Given a
multiparticle quantum state, we first show that Gilbert's algorithm can be
adapted to prove separability or membership in a certain entanglement class.
We then present two algorithms for convex optimization over SLOCC classes.
The first algorithm uses a simple gradient approach,  while the other one
employs the accelerated projected-gradient method. For demonstration, the
algorithms are applied to the likelihood-ratio test using experimental data
on bound entanglement of a noisy four-photon Smolin state
[Phys. Rev. Lett. \textbf{105}, 130501 (2010)].
\end{abstract}

\pacs{03.65.Ta, 03.65.Ud}

\maketitle

\textit{Introduction.---}
Entanglement is a fundamental phenomenon in quantum mechanics and is
often considered to be a useful resource for tasks like quantum metrology
or quantum cryptography. Consequently, the question of whether a given quantum
state of two particles is entangled or separable is relevant for several
fields in physics \cite{ent1,ent2}. So far, much effort has been devoted to
devise methods to certify that a given state is entangled, prominent examples
are Bell inequalities and entanglement witnesses \cite{ent2, EW}. Contrary to
that, methods to prove separability, or, equivalently, to disprove the presence
of entanglement, are rare: For some instances explicit solutions
are known \cite{prl80.2245, pra58.826, jpa40.483, pra83.020303}, but in general
one has to rely on numerical procedures with a restricted applicability
\cite{pra76.032318, prl103.160404,
np6.943, pra86.032307}. In order to analyze experiments, one often
needs to quantify the extent to which the observed data can be explained by
separable states, e.g., in a likelihood-ratio test \cite{prl105.170501}. In this
case, one even needs to {\it optimize} over separable states and this is nearly impossible
with current analysis tools. The analysis of separability becomes even more complicated
in the multiparticle case, because then different classes of entanglement exist
\cite{slocc1,slocc2}.

In this Letter, we present methods to analyze separability for multiparticle quantum
states. First, we show that an adaption of the so-called Gilbert's algorithm \cite{Gilbert}
can be used to prove separability or membership in a certain entanglement class, and the
resulting algorithm outperforms known methods significantly. Second, we demonstrate that
a combination of this method with gradient methods can be used to perform optimization
over separability classes, allowing, for instance, the computation of likelihood ratios.
We demonstrate the practical usefulness of our approach with many examples and data
from recent experiments.


\textit{Notions of entanglement.---}
A pure state $\ket{\psi}$ of two particles is separable, if it is a product
state $\ket{\psi}={\ket{a}\otimes \ket{b}}$, otherwise it is entangled. Concerning
mixed states, a state is separable if it can be written as a convex combination
of product states, that is
\begin{equation}
\rho = \sum_i p_i \ket{a_i}\bra{a_i} \otimes \ket{b_i}\bra{b_i}\,.
\label{eq:convexsum}
\end{equation}
Here, the $p_i$ form a probability distribution, so they are positive and sum
up to one. A state that is not separable is called entangled. While the entanglement
of pure states is straightforward to characterize, the same question for mixed
states is a hard problem \cite{ent1,ent2}.

For more than two particles, different classes of entanglement exist. To give
a first example, let us consider ${n > 2}$ particles and fix a number ${2 \leq k \leq n}$.
A pure state is then $k$-separable if it can be written as a tensor product of
$k$ local states,
\begin{equation}
\ket{\Phi_{\textsc{S}_{k}}} = \bigotimes_{i=1}^{k}\ket{\phi_i}\,,
\end{equation}
where the $\ket{\phi_i}$ are $k$ states on subsets of the $n$ parties.
More specifically, the state are called biseparable for $k=2$, triseparable
for $k=3$, and up to fully separable for $k=n$. A state that is not $k$-separable
contains entanglement, for instance, it is genuinely multiparticle entangled
if it is not biseparable.
For mixed states, one can extend this definition as before via convex combinations.
In this case, one also considers mixtures of $k$-separable states that are separable
with respect to different partitions. It is known that the characterization
of mixed separable states is $NP$ hard, if the number of particles increases \cite{Gurvits2003, Gharibian2010}.

To characterize multiparticle entanglement further, a popular strategy uses
the notion of stochastic local operations and classical communication (SLOCC)
\cite{slocc1,slocc2}. In mathematical terms, a SLOCC operation can be
represented as $A_{\textsc{SLOCC}}=\bigotimes_{i}A_i$, where $A_i$ is
a matrix describing the local operation acting on the $i$th party. Under
SLOCC operations, a pure state $\ket{\phi}$ can be mapped to another state
$\ket{\phi'}$ iff
\begin{equation}
\label{eq:slocc}
\ket{\phi'}\propto A_{\textsc{SLOCC}}\,\ket{\phi}\,,
\end{equation}
and $\ket{\phi}$ and $\ket{\phi'}$ are called SLOCC equivalent if
the local operations $A_i$ are invertible \cite{slocc1}. Remarkably, for
three qubits this classification gives two inequivalent families of genuine
multiparticle states, the Greenberger-Horne-Zeilinger class and the $W$ class \cite{slocc1}. Again,
one can define the corresponding convex sets for mixed states as in
Eq.~(\ref{eq:convexsum}) \cite{slocc2}. We denote such a SLOCC
entanglement class by $\cC$.


\textit{Membership in SLOCC classes.---}
How can one determine whether a given quantum state is separable or belongs to
any specific SLOCC entanglement class $\cC$? In principle the algorithm
introduced in Refs.~\cite{np6.943, pra86.032307} is applicable, but no
convergence can be guaranteed, and the algorithm fails in general
for rank-deficient states.

Recently, Brierley, Navascu\'es, and V\'ertesi \cite{arXiv1609.05011} presented a scheme
for the problem of convex separation based on the so-called Gilbert's algorithm
\cite{Gilbert}. This scheme is shown to outperform existing linear programming
methods for certain large scale problems in quantum information theory. For
instance, nonlocality in bipartite scenarios can be certified with up to 42
measurement settings, new upper bounds are obtained for the visibility of
certain states, as well as the steerability limit of Werner states. Basically,
given any quantum state $\rho$, Gilbert's algorithm searches for a state
$\rho^{\cC}\in\cC$ which approximates the minimal distance between $\rho$
and the convex set $\cC$. We denote such an operation by applying Gilbert's
algorithm as ${\rho^{\cC}\equiv\cS(\rho)}$ for later use, see also
Appendix~A \cite{SM} for detailed discussions about this algorithm.

In any case, Gilbert's algorithm searches for an approximation of a given
state $\rho$ within the convex set $\cC.$ If a good approximation is found,
this does not mean that the state $\rho$ is within the set, as still it may be outside,
but close to the boundary. Nevertheless, using some facts from the entanglement theory,
we can modify the algorithm:

\begin{proposition}\label{prop1}
Gilbert's algorithm can be adapted to prove separability or membership in a certain SLOCC entanglement
class $\cC$.
\end{proposition}
\begin{proof}
{The proof relies on two facts about $\cC$: (i) convexity, and (ii) highly mixed states
belong to $\cC$. If the state $\rho$ to be checked can be written as a convex combination
of the state found by Gilbert's algorithm and any state within the highly mixed region,
then ${\rho\in\cC}$. See Appendix~B \cite{SM} for the complete proof.}
\end{proof}

By making use of Proposition~\ref{prop1}, we tested different types of
entanglement for various multiparticle quantum systems in
Appendix~C \cite{SM}. It clearly shows that in most of the cases,
Proposition~1 gives much better results compared to those obtained
from previous known methods. Whereas, in the following, we use
Gilbert's algorithm as a tool to ensure constraints.

\textit{Convex optimization.---}
Denote by $\cF(\rho)$ a strictly concave
(or convex) function defined over the quantum state space $\cQ$, such
that $\cF(\rho)$ has a single maximum (or minimum). Many functions in
quantum information science meet this requirement, for instance the
log-likelihood function, the von Neumann entropy, etc. The statistical
operator $\rho\in\cQ$ has to satisfy two constraints, namely,
\begin{equation}
\rho\ge 0\,\,\textrm{and}\,\,\tr(\rho)=1\,.\label{eq:constr}
\end{equation}
We also assume that $\cF(\rho)$ is differentiable (except perhaps at a few
isolated points) with gradient $\nabla{\cF}(\rho)\equiv G(\rho)$. The
objective is to maximize $\cF(\rho)$ over a specific SLOCC entanglement
class $\cC\subseteq\cQ$, i.e., a convex subset of the state space.
Explicitly, we have
\begin{subequations}\label{eq:prob}
\begin{eqnarray}
\textrm{maximize}&\quad&\,\,{\cF}(\rho)\,,\\
\textrm{subject to}&\quad&\,\,\rho\in\cC\,.\label{eq:constrS}
\end{eqnarray}
\end{subequations}
We denote the solution of this optimization  by $\rhoML^{\cC}$.

As mentioned, it is hard to test whether a given state belongs to
a SLOCC entanglement class, which makes the optimization
defined above even harder. For this problem, we offer two iterative schemes,
where the constraint in Eq.~\eqref{eq:constrS} is guaranteed
by Gilbert's algorithm. Specifically, each iterative step
involves two operations, namely one gradient operation (the update)
followed by one Gilbert's operation (constraints enforced). For
the gradient, we have two different
approaches.

\textit{The direct-gradient (DG) scheme.---}
Let us first consider the case
when ${\cC=\cQ}$, then the constraint in Eq.~\eqref{eq:constrS} is identical
to that of Eq.~\eqref{eq:constr}, which can be ensured if one writes $\rho=A^{\dag}A/\tr(A^{\dag}A)$. In the unconstrained $A$ space, the small variation of $\cF(\rho)$ is given by
\begin{equation}
\delta{\cF}(\rho)\equiv\delta{\cF}(A)=\tr\Big(\delta A\frac{\bigl[G-\tr(G\rho)\bigr]A^{\dag}}{\tr(A^{\dag}A)}+H.c.\Big),
\end{equation}
to linear order in $\delta A$. If we choose $\delta A=\epsilon A\bigl[G-\tr(G\rho)\bigr]$ with $\epsilon$ being positive, then $\delta{\cF}(\rho)$ is always positive, hence walking upwards. Thus, by following the gradient, we have the update for $\rho$ in the DG scheme as
\begin{eqnarray}
\rho_{k+1}&=&\frac1{\cal N}\bigl(\openone+\epsilon\bigl[G_k-\tr(G_k\rho)\bigr]\bigr)\rho_{k}\bigl(\openone+\epsilon\bigl[G_k-\tr(G_k\rho)\bigr]\bigr)\,,\nonumber\\
&\equiv&\mbox{DG}(\rho_k,G_k,\epsilon)\,.\label{eq:dg}
\end{eqnarray}
with ${\cal N}$ being the normalization constant.

Once the iteration is finished, the algorithm returns the optimal quantum state $\rhoML$ with the corresponding optimal function value ${{\cF}_{\textsc m}={\cF}(\rhoML)}$ over the whole quantum state space. When $\cC$ is strictly smaller than $\cQ$, the state after the update in Eq.~\eqref{eq:dg} may easily be outside of $\cC$. Whenever this happens, we use Gilbert's algorithm to project $\rho_k$ back to $\cC$, i.e., $\rho_k\rightarrow\rho_k^{\cC}=\cS(\rho_k)$. Note that we also assume $\rhoML\notin\cC$, otherwise $\rhoML^{\cC}\equiv\rhoML$ then the optimization in Eq.~\eqref{eq:prob} is solved. With all the ingredients at hand, the DG algorithm proceeds as follows:
\begin{algorithm}[H]
\caption{{\bf DG}}
\begin{algorithmic}
\vspace*{0.1cm}
\State Given $\epsilon>0$ and $0<\beta<1$.
\State Choose any $\rho_0^{\cC}\in\cC$, and ${\cF}_{0}={\cF}(\rho_0^{\cC})$.
\vspace*{0.1cm}
\For {$k = 1,\cdots,$}
\vspace*{0.1cm}
\State Update $\rho_{k}=\mbox{DG}\!\left[\rho_{k-1}^{\cC},G\!\left(\rho_{k-1}^{\cC}\right)\!,\epsilon\right]$.
\State Calculate $\rho_{k}^{\cC}=\cS(\rho_{k})$, and ${\cF}_{k}={\cF}(\rho_k^{\cC})$.
\vspace*{0.1cm}
\State Termination criterion!
\vspace*{0.1cm}
\If {${\cF}_{k}<{\cF}_{k-1}$} \hspace*{0.1cm} (No update)
\State Reset $\epsilon=\beta\epsilon$ and $\rho_k^{\cC}=\rho_{k-1}^{\cC}$.
\EndIf
\vspace*{0.1cm}
\EndFor
\end{algorithmic}
\end{algorithm}
The initial step-size $\epsilon$ can be chosen rather arbitrarily, which does not affect the final output
too much. Though the DG algorithm is very simple and straightforward to use,
it suffers two problems: slow convergence and low precision. These are due
to the fact that the iterations in DG are actually performed in the
unconstrained $A$ space. When $\rho$ is close to
the boundary of the state space, it eventually becomes rank-deficient
with at least one small eigenvalue. The highly-asymmetric spectrum would
cause the gradient to be locally ill-defined \cite{apg}. To avoid
these problems the state-of-the-art optimization method is to walk
directly in the $\rho$ space.

\textit{The accelerated projected-gradient (APG) scheme.---}
The APG approach \cite{FISTA, AdaptRes, TFOCS} is generally applicable in all kinds of constrained problems where the constrains are enforced by a projection operation \cite{Goldstein, Levitin, Bruck, Passty, Nesterov}. In the current scenario, we have to make sure that the update for $\rho$ at each iterative step stays in $\cC$ all the time, for which we use Gilbert's algorithm. Rather different from common gradient approaches, update of the target $\rho$ in APG is based on another state $\sigma$, such that each update gets some ``momentum" from the previous step in order to find the optimal direction for update. The momentum is controlled by $\theta$ in the algorithm, which will be reset to 1 whenever it causes the current step to point too far from the DG direction. Upon convergence, $\rho$ and $\sigma$ will eventually merge to the same point. For more technical details about the APG algorithm, e.g., the `Restart' and `Accelerate' operations, we refer to Ref.~\cite{apg}. Then the APG algorithm proceeds
as follows:
\newpage
\begin{algorithm}[H]
\caption{{\bf APG}}
\begin{algorithmic}
\vspace*{0.1cm}
\State Given $\epsilon>0$ and $0<\beta<1$.
\State Choose any $\rho_0^{\cC}\in\cC$, and ${\cF}_{0}={\cF}(\rho_0^{\cC})$;
\State \hspace*{1.74cm}$\sigma_0=\rho_0^{\cC}$, and $\theta_0=1$.
\vspace*{0.1cm}
\For {$k = 1,\cdots,$}
\vspace*{0.1cm}
\State Update $\rho_k^{\cC}=\cS[\sigma_{k-1}+\epsilon\,G(\sigma_{k-1})]$, ${\cF}_{k}={\cF}(\rho_k^{\cC})$.
\vspace*{0.1cm}
\State Termination criterion!
\vspace*{0.1cm}
\If {${\cF}_{k}<{\cF}_{k-1}$} \hspace*{0.1cm} (Restart)
\State Reset $\epsilon=\beta\epsilon$, $\rho_k^{\cC}=\rho_{k-1}^{\cC}$, $\sigma_k=\rho_k^{\cC}$, and $\theta_k=1$.
\Else \hspace*{0.1cm}(Accelerate)
\State Set $\theta_k=\tfrac{1}{2}{\left(1+\sqrt{1+4\theta_{k-1}^2}\right)}$;
\State Update $\sigma_k=\rho_k^{\cC}+\frac{\theta_{k-1}-1}{\theta_k}\left(\rho_k^{\cC}-\rho_{k-1}^{\cC}\right)$.
\EndIf
\vspace*{0.1cm}
\EndFor
\end{algorithmic}
\end{algorithm}

The \textsc{matlab} codes for the DG and APG algorithms, with accompanying documentation and implementations, are available online \cite{code}.
To guarantee the validity of these two algorithms, we have the following theorem.
\begin{theorem}\label{theo1}
Suppose Gilbert's algorithm is precise, i.e., the operation ${\rho^{\cC}\equiv\cS(\rho)}$
always returns the closest $\rho^{\cC}$  with respect to $\rho$ in the Hilbert-Schmidt norm.
Then, if the iteration reaches a fixed point by the DG algorithm or the APG algorithm,
this point is the solution to the optimization in Eq.~\eqref{eq:prob}.
\end{theorem}
\begin{proof}
{We prove this theorem by assuming contradictions: the convexity properties of $\cF(\rho)$ and
$\cC$ are not compatible with two different solutions. For the complete proof, see Appendix~D \cite{SM}.}
\end{proof}

\textit{The likelihood-ratio test.---}
In real-world experiments, resources are limited, thus the data obtained are always finite. Drawing conclusions
from a finite amount of data requires statistical reasoning. In Ref.~\cite{prl105.170501}, a universal method for quantifying the weight of evidence for (or against) entanglement with finite data was introduced. However, being boiled down to an optimization over specific convex sets of entanglement classes, this method is generally not doable. Here, we show that this problem can be tackled by using our algorithms.

In a typical quantum tomographic scenario \cite{LNP649}, $N$ independently and identically prepared copies of the quantum state $\rho$ are measured by a positive operator-valued measure (POVM) $\{\Pi_k\}_{k=1}^K$, with ${\Pi_k\ge 0}\,\,\forall k$ and ${\sum_{k=1}^K\Pi_k=\openone}$. The data ${D=\{n_1,n_1,\dots,n_K\}}$ consist of a sequence of detector clicks, the probability of getting which is given by the likelihood function
\begin{equation}
\cL(D|\rho)=\prod_{k} p_{k}^{n_k}={\left\{\prod_{k}{\left[\tr(\rho\Pi_k)\right]}^{f_k}\right\}}^N\,,
\end{equation}
where ${p_k=\tr(\rho\Pi_k)}$ (Born rule) is the probability for outcome $\Pi_k$, and $f_k=n_k/N$ denotes the relative frequency. Note that $\cL(D|\rho)$ is not strictly concave, but the normalized log-likelihood $\cF(\rho)\equiv\frac{1}{N}\ln\cL(D|\rho)$ is.

The likelihood ratio in Ref.~\cite{prl105.170501} is defined as
\begin{equation}\label{eq:lratio}
  \Lambda\equiv\frac{\max_{\rho\in \cC} {\cL}(D|\rho)}{\max_{\mbox{\scriptsize all}\,\rho}{\cL}(D|\rho)}\,
\mbox{ and }
  \lambda=-2\ln(\Lambda)
\end{equation}
represents the weight of evidence in favor of entanglement. Hence, to demonstrate entanglement convincingly, a large value of $\lambda$ is demanded. Moreover, for states lying close to the boundary of $\cC$, it has been shown in \cite{prl105.170501} that $\lambda$ follows a semi-$\chi_1^2$ distribution for large enough $N$. By having this, one can perform hypothesis testing to demonstrate entanglement, then construct confidence levels. Suppose we get $\rho_{\scriptsize\mbox{exp}}$ with the corresponding $\lambda_{\scriptsize\mbox{exp}}$ in an experiment, the $p$-value for the null hypothesis that ${\rho_{\scriptsize\mbox{exp}}\in\cC}$ is given by the probability ${\mbox{Pr}(\lambda>\lambda_{\scriptsize\mbox{exp}})\equiv\epsilon}$. Therefore, with the $(1-\epsilon)$ confidence level, the null hypothesis has to be rejected, indicating that the state is entangled.

The likelihood ratio defined in Eq.~\eqref{eq:lratio} involves two optimizations over two different convex sets. The maximization in the denominator is well known as the maximum-likelihood estimation \cite{MLEreview}, for which various algorithms exist; while the maximization in the numerator fits exactly into our problem. Let us denote the solutions to these two maximizations by ${\rhoML=:\arg \max_{\mbox{\scriptsize all}\,\rho} {\cF}(\rho)}$ and ${\rhoML^{\cC}=:\arg \max_{\rho\in \cC} {\cF}(\rho)}$ respectively, then Eq.~\eqref{eq:lratio} can be rewritten as
$\lambda=2N\!\big[\cF(\rhoML)-\cF(\rhoML^{\cC})\big].$

\begin{figure}[]
\includegraphics[width=0.8\columnwidth]{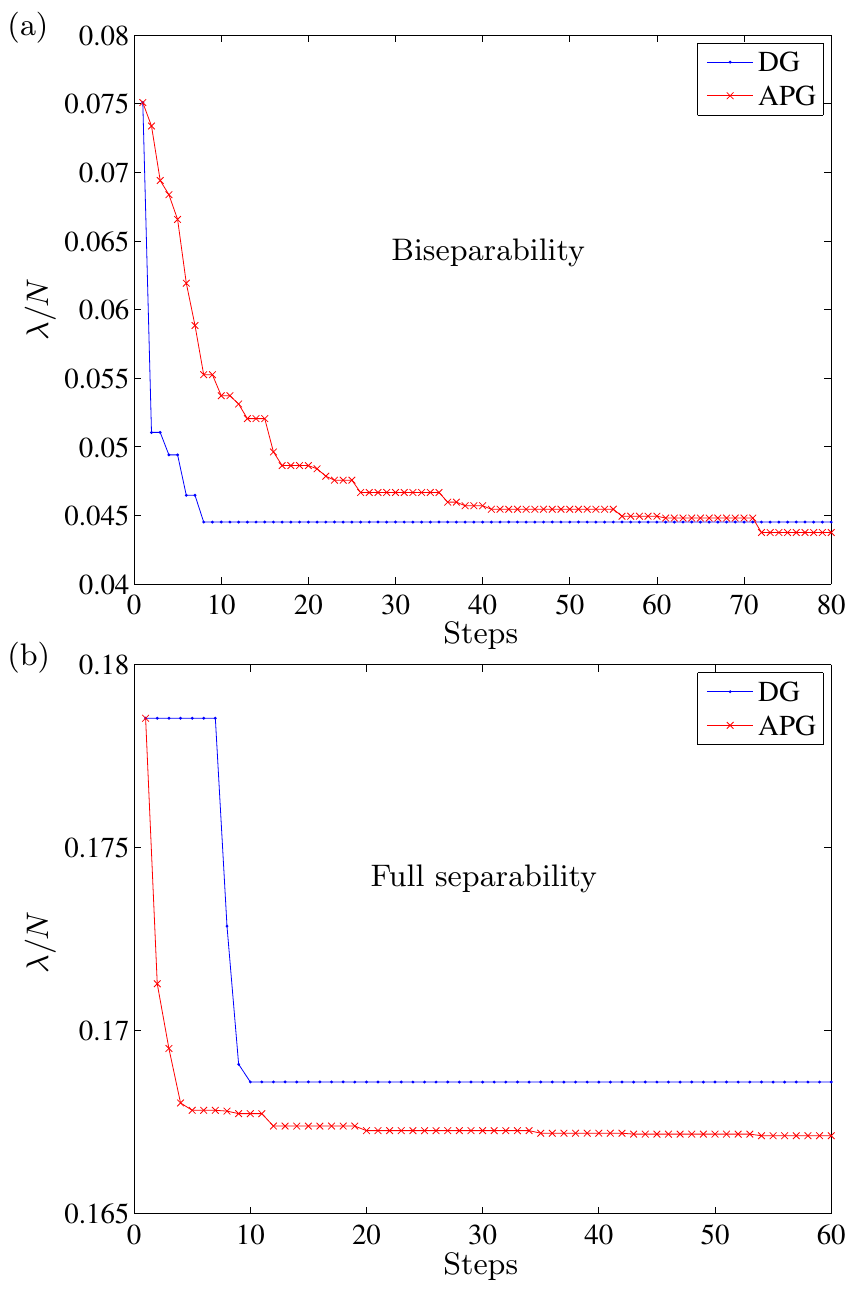}
\caption{\label{fig:W4}
The normalized log-likelihood ratios $\lambda/N$ at each iterative step by DG and APG respectively, for the four-qubit $W$ state with white noise: (a) Biseparability; (b) Full separability.
The plateaus in the plots imply the process where the algorithms are searching for suitable step-sizes for the next update. As shown, the APG algorithm usually returns more accurate solutions than DG does.}
\end{figure}
\textit{Four-qubit $W$ state with white noise.---}
For the first application, let us consider the four-qubit $W$ state,
$\ket{W_4}=(\ket{0001}+\ket{0010}+\ket{0100}+\ket{1000})/2$
mixed with white noise,
\begin{equation}\label{eq:rhop}
\rho_{W_4}(q)=q\ket{W_4}\bra{W_4}+\frac{1-q}{16}\openone\,.
\end{equation}
By employing Proposition~1, we find $\rho_{W_4}(q)$ is biseparable for ${q\le 0.4555}$ and fully separable for ${q\le0.09}$; see Table~I in Appendix~C \cite{SM}.

In the simulation, we choose the noise level ${q=0.9}$, then employ the standard Pauli tomographic scheme where each qubit is measured in the basis of the three Pauli operators. Without the loss of generality, we set ${\{f_k=p_k\}}$ such that the maximum-likelihood estimator \emph{is} the true state. Hereafter, we calculate instead the normalized log-likelihood ratios, i.e., $\lambda/N$. Figure~\ref{fig:W4} shows the results for testing biseparability as well as full separability for this case. As expected, the $\lambda/N$ value obtained for biseparability is much smaller than that for full separability, as the fully separable states consist of a strictly smaller subset of the biseparable region. Moreover, the APG algorithm usually has a better precision-resolvent capability than DG does. For more simulated examples, see Appendix~E \cite{SM}.

\textit{Experimental bound entanglement.---}
The four-party Smolin state \cite{Smolin} is
\begin{equation}
\rho_{\mathrm S}=\frac1{4}\sum_{\mu=1}^{4}\ket{\Psi^{\mu}}\bra{\Psi^{\mu}}_{AB}\otimes\ket{\Psi^{\mu}}\bra{\Psi^{\mu}}_{CD}\,,
\end{equation}
where the subscripts label the parties and $\ket{\Psi^{\mu}}$ are the two-qubit Bell states. By adding white noise, we have
$
\rho_{\mathrm S}(q)=q\rho_{\mathrm S}+({1-q})\openone/16,
$
which is fully separable for ${q\le 1/3}$, and bound entangled for ${q>1/3}$ \cite{pra74.010305}.
In Ref.~\cite{PhotonSmo}, a family of noisy four-photon Smolin states was generated by spontaneous parametric down-conversion. By varying the noise level, bound entanglement was successfully demonstrated for ${q=0.51}$. Here, we re-analyze their experimental data using the likelihood-ratio test.

To demonstrate bound entanglement, one has to show that the state has a positive partial transpose (PPT) \cite{PPT}, but is nevertheless entangled. For this, optimizations over two different convex sets, namely, sets of the fully separable states as well as the PPT states, have to be performed; see the results in Fig.~\ref{fig:exp}. At noise level $q=0.51$, we get ${\lambda\approx2.42\times10^{3}}$ with the {$p$-value $\approx0$} for the null hypothesis that the state is separable. Thus, the null hypothesis has to be rejected, so the state is indeed entangled. Meanwhile, we get ${\lambda\approx1.94\times10^{-6}}$ for the optimization over PPT states, indicating strongly that the state is PPT \cite{note1}. Therefore, the state at noise level $q=0.51$ is both entangled and PPT, thus bound entangled. Similarly, one can conclude from the $\lambda$ values that the state at noise level $q=0$ is separable, while the state at $q=1$ is genuinely entangled.

In Appendix~F \cite{SM}, we use simulated data to perform the likelihood-ratio test for various noise levels. By doing so, we identify the parameter range ${q\sim [0.35, 0.8]}$ (containing ${q=0.51}$ from the real experiment), which is most likely to show bound entanglement.
\begin{figure}[]
\includegraphics[width=0.8\columnwidth]{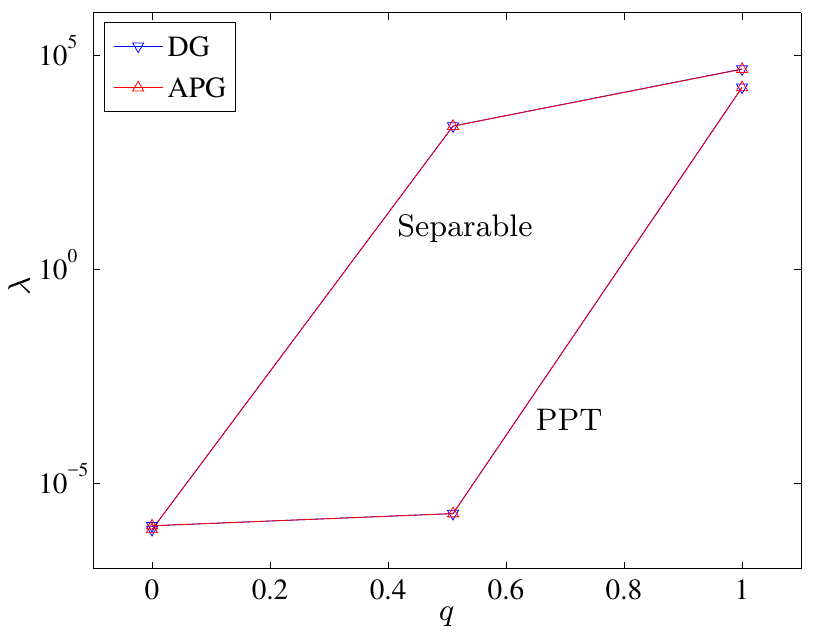}
\caption{\label{fig:exp}
Likelihood-ratio test for bound entanglement for the experimental four-photon Smolin state \cite{PhotonSmo}. The top curve labeled ``Separable" indicates the $\lambda$ values obtained via maximizing over the fully separable states; while the bottom curve labeled ``PPT" gives values obtained by maximizing over the PPT states. Hypothesis testing suggests that the state at noise level $q=0.51$ is both entangled and PPT, thus bound entangled.
}
\end{figure}
%

\textit{Conclusions.---}
The characterization of multiparticle entanglement is generally hard. In this work, we show that Gilbert's algorithm can be adapted to prove a given quantum state is either separable or belongs to a SLOCC entanglement class, with the thresholds thus obtained being much better than those reported by previous known methods. Furthermore, with the help of Gilbert's algorithm, two reliable schemes are presented for the convex optimization over any defined SLOCC entanglement classes. For demonstration, we re-analyzed the experimental data on bound entanglement of the noisy four-photon Smolin states using the likelihood-ratio test. As such, we expect that our methods would become a reliable tool for experimentalists to test the entanglement property of their quantum systems with confidence.

%
This work has been supported by the ERC (Consolidator Grant No. 683107/TempoQ), and the DFG. We thank H. K. Ng, Z. Zhang, S. Brierley, T. V\'ertesi, and H. Zhu for stimulating discussions. We are also grateful to J. Lavoie for sharing the experimental data in Ref.~\cite{PhotonSmo}, to M. Kleinmann and T. Monz for sharing the data in Ref.~\cite{np6.943}, and to H. Kampermann for sharing the codes used in Ref.~\cite{pra86.032307}.

\section{Appendix A: Gilbert's algorithm}\label{sec:Gil}
The scheme presented in Ref.~\cite{arXiv1609.05011} is based on Gilbert's
algorithm for quadratic minimization \cite{Gilbert}. Here, we discuss this
scheme using the language in the present context, where we extend it to be applicable for any defined SLOCC entanglement classes. For more technical
details about Gilbert's algorithm we refer the reader to Refs.~\cite{arXiv1609.05011, Gilbert}.

The optimization problem that Ref.~\cite{arXiv1609.05011} solves is the so-called \emph{Weak minimum Distance (WDIST)}: Given any quantum state $\rho$, Gilbert's algorithm searches for a state $\rho^{\cC}\in\cC$, such that
\begin{equation}\label{eq:Gil}
||\rho-\rho^{\cC}||\le\mbox{dist}({\cC},\rho)+\delta\,,
\end{equation}
where $\mbox{dist}({\cC},\rho)$ denotes the minimal distance between $\rho$ and the convex set $\cC$ in Hilbert-Schmidt (HS) norm, and $\delta$ is a pre-defined tolerance. Specifically, Gilbert's algorithm with memory $m$ proceeds as follows:
\begin{algorithm}[H]
\caption{{\bf Gilbert with memory}}
\begin{algorithmic}
\vspace*{0.1cm}
\State Choose any $\rho_1^{\cC}\in\cC$.
\vspace*{0.1cm}
\For {$k = 1,\cdots,$}
\vspace*{0.1cm}
\State 1. Use an oracle to solve $\sigma_k=:\arg\max_{\sigma\in\cC}\,\bigl[(\rho-\rho_k^{\cC})\cdot\sigma\bigr]$.
\State $\quad\,$Append $\sigma_k$ into $A$ as a column. $\quad$ (Memory)
\vspace*{0.1cm}
\State 2. Solve $\bar{x}_{min}=:\arg\min_{\bar{x}\succ0,\sum\bar{x}=1}||A\bar{x}-\rho||$.
\State $\quad\,$Update $\rho_{k+1}^{\cC}\equiv A\bar{x}_{min}$.
\vspace*{0.1cm}
\State Termination criterion!
\vspace*{0.1cm}
\EndFor
\end{algorithmic}
\end{algorithm}
\begin{figure}[]
\includegraphics[width=0.8\columnwidth]{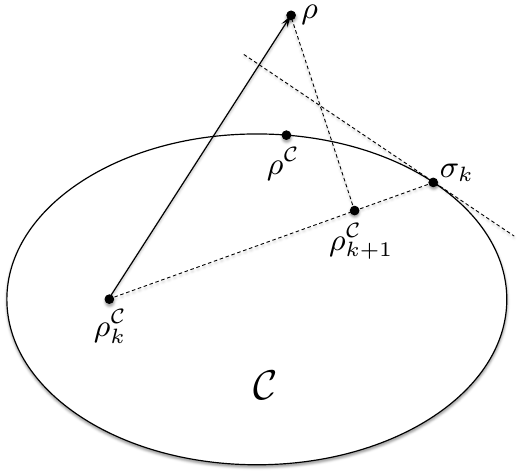}
\caption{\label{fig:Gil}%
Geometrical description of Gilbert's algorithm with memory $m=1$.
The state $\rho$ should be approximated and a state $\rho^\cC_k$
in $\cC$ is known. One first computes $\sigma_k$, and then defines
$\rho^\cC_{k+1}$ as the best approximation to $\rho$ on the line
connecting $\rho^\cC_k$ with $\sigma_k$. This procedure can be
iterated, see the text for details.
}
\end{figure}
Figure~\ref{fig:Gil} shows a geometrical description of Gilbert's algorithm when the memory ${m=1}$. More memory is better for convergence, but would cost more time for each iteration. To balance the trade-off, in practice we usually set
${m=50}$. For the maximization in step 1, we adopt a heuristic oracle as
that used in Refs.~\cite{np6.943, pra86.032307}. First, instead of considering
the whole convex set $\cC$,  it is sufficient to optimizate over pure states
only, such that
\begin{equation}
\sigma_k=:\arg\max_{|\Phi'\rangle\in\cC}\bra{\Phi'}(\rho-\rho_k^{\cC})\ket{\Phi'}\,.
\end{equation}
where ${\ket{\Phi'}\propto A_{\textsc{SLOCC}}\,\ket{\Phi_0}}$ with arbitrary initial ${\ket{\Phi_0}\in\cC}$. Then, one can perform this optimization
iteratively, where in each step $n-1$ of the single-particle transformations
$A_i$ are fixed, while the
remaining one can be determined analytically \cite{np6.943, pra86.032307}.
Note that a certified optimal solution is not needed in this step as long
as the returned $\sigma_k$ stays in $\cC$. The block matrix $A$ contains
$m$ entries, whereas extra ones (those earliest added) should be erased. The minimization in step 2 is equivalent to projecting $\rho$ onto the line ${\vec{l}=\sigma_k-\rho_k^{\cC}}$, which is a simple linear constraint problem thus can be solved easily. Since the projected $\rho_{k+1}^{\cC}$ is a convex combination of two states within $\cC$, the update in Gilbert's algorithm is guaranteed to stay in $\cC$. After a finite number of iterations, the good approximation ${\rho^{\cC}\in\cC}$, satisfying the inequality in Eq.~\eqref{eq:Gil} is returned.

\section{Appendix B: Proof of Proposition~1}\label{sec:prop1}
\begin{figure}[]
\includegraphics[width=0.75\columnwidth]{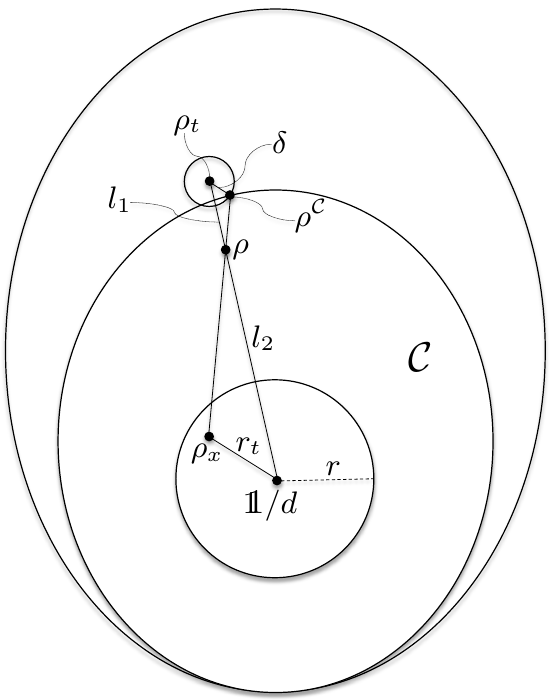}
\caption{\label{fig:prop1}%
Geometrical description of the proof for Proposition~1. See the text
for details.
}
\end{figure}
As mentioned in the main text, Gilbert's algorithm cannot be used directly to certify separability nor prove whether a given quantum state belongs to a SLOCC entanglement class or not. Let us look at Eq.~\eqref{eq:Gil} in the \emph{WDIST} definition. Suppose ${\rho\in\cC}$, we would have $\mbox{dist}({\cC},\rho)=0$, then $||\rho-\rho^{\cC}||\le\delta$. Thus, if $||\rho-\rho^{\cC}||>\delta$, one can conclude that $\rho\notin\cC$. However, ambiguity may happen if $\rho\notin\cC$, but lies very close to the boundary of $\cC$ (namely when ${\mbox{dist}({\cC},\rho)<\delta}$), then the result may be interpreted in the wrong way.

To prove Proposition~1, we need one fact about the convex set $\cC$, that is, highly mixed states belong to $\cC$. This, of course, depends on the structure of $\cC$. For example, consider bipartite $d_1\times d_2$ systems and let $\cC$ denote the set of separable states, it has been shown that if
\begin{equation}
\tr(\rho^2)\le\frac1{d_1d_2-1}\,,
\end{equation}
then $\rho$ is separable, i.e., $\rho\in\cC$ \cite{pra66.062311}. In terms of the HS norm, we have a finite region surrounding the completely mixed state $\openone/d$ with radius $r=1/\sqrt{d(d-1)}$ such that all the states contained are separable; see Fig.~\ref{fig:prop1}. Similar results have been obtained for other SLOCC entanglement classes \cite{np6.943}.

Given any multiparticle quantum state $\rho$, Proposition~1 says that Gilbert's algorithm can be adapted to prove $\rho$ is either separable or belongs to a SLOCC class $\cC$. Firstly, choose a small real positive value $\epsilon$, and construct the following state
\begin{equation}
\rho_t=(1+\epsilon)\rho-\frac{\epsilon}{d}\openone\,.
\end{equation}
Next, we run Gilbert's algorithm to find the closest state ${\rho^{\cC}\in\cC}$ with respect to $\rho_t$; see Fig.~\ref{fig:prop1}. We can connect
$\rho^{\cC}$ and $\rho$ then extrapolate to $\rho_x$,
where $\rho_x$ is defined by the condition that the two lines indicated by $r_t$
and $\delta$ are parallel. From a geometrical perspective, we have
\begin{equation}
  \frac{\delta}{r_t}=\frac{l_1}{l_2}\quad\Rightarrow\quad r_t=\frac{l_2}{l_1}\delta\,.
\end{equation}
Thus, if $r_t\le r$, then $\rho_x$ is in $\cC$, and so we certify that $\rho$ is separable or $\rho\in\cC$ since it is the convex combination of two states within $\cC$. Finally, if $r_t>r$, one can try to repeat the process by varying
$\epsilon$, e.g., decreasing it until $\epsilon < \delta$, i.e., the numerical tolerance that we set for Gilbert's algorithm.

\section{Appendix C: Applications of Proposition~1}\label{sec:thres}
\begin{figure}[]
\includegraphics[width=0.8\columnwidth]{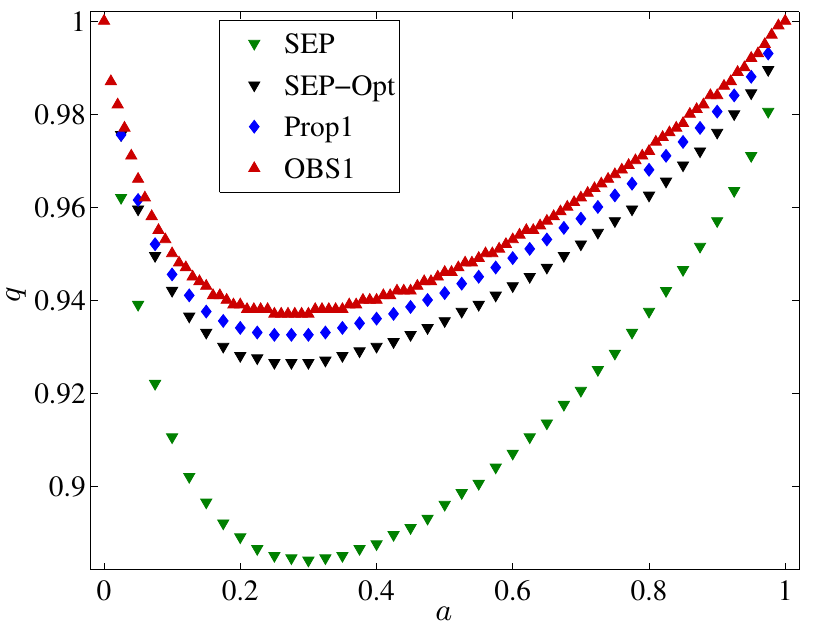}
\caption{\label{fig:sep_PH}%
Separability test for the Horodecki $3\times3$ bound entangled states mixed with white noise. ``SEP" denotes the algorithm introduced in Ref.~\cite{pra86.032307}, and ``OBS1" represents the entanglement criterion reported in Ref.~\cite{prl109.200503}. As can be seen, the values found by Proposition~1 (``Prop1") are significantly improved over those by ``SEP", and tend to reach the bound by ``OBS1". Note that the optimized values by ``SEP" (denoted by``SEP-Opt'') are still worse than the values obtained by Proposition~1.
}
\end{figure}
\setlength{\tabcolsep}{0.85em} 
\renewcommand{\arraystretch}{1.2}
\begin{table}[]
\begin{tabular}{lccll}
\hline\hline
State   &   Ent.   &   Bound   &   SEP \cite {pra86.032307}   &   Prop1  \\ \hline
$\rho_{\mathrm{GHZ} 3}$   &   S   &   $1/5^{a,b}$   &   $0.199$   &   $0.1986$ \\
        &   BS   &   $0.429^{a}$ \cite{prl106.190502}   &   $0.4285$   &   $0.4281$ \\
        &   W   &   $0.6955^{a}$ \cite{prl108.020502}   &   $0.694$   &   $0.6948$ \\ \hline
$\rho_{\mathrm{GHZ} 4}$   &   S   &   $1/9^{a,b}$   &   $0.111$   &   $0.1087$ \\
        &   BS   &   $0.467^{a}$ \cite{prl106.190502}   &   $0.466$   &   $0.4609$ \\
        &   W   &      &   $0.316$   &    $0.3512$ $\uparrow$  \\ \hline
$\rho_{\mathrm{W3}}$   &   S   &   $3/11^{b}$   &   $0.1727$   &   $0.177$  \\
        &   BS   &   $0.479^{a}$ \cite{prl106.190502}   &   $0.45$   &   $0.4745$ $\uparrow$  \\ \hline
$\rho_{\mathrm{W4}}$   &   S   &   $1/5^{b}$   &   $0.09$   &   $0.09$ \\
        &   BS   &   $0.474$ \cite{prl106.190502}   &   $0.434$   &   $0.4555$ $\uparrow$  \\ \hline
$\rho_{\mathrm{UPB}}$   &   S   &   $0.87$ \cite{prl88.187904}   &   $0.83$   &   $0.863$ $\,\,\,\uparrow$  \\ \hline
$\rho_{\mathrm{BE3}}$   & S   &   $0.786^{a}$ \cite{HyllusThs}   &   $0.726$   &   $0.732$  \\
        &   BS   &   $1^{a}$ \cite{slocc2}   &   $0.9$   &   $0.9985$ $\uparrow$  \\
\hline\hline
\multicolumn{5}{l}{$^a$Exact values from the literature.}  \\
\multicolumn{5}{l}{$^b$Bounds obtained via the PPT criterion.}
\end{tabular}
\caption{\label{tab:thres}
Threshold values $q$ for different types of entanglement for various multiparticle quantum states $\rho$ mixed with white noise, i.e., $\rho(q)=q\rho+(1-q)\openone/d$. The column ``SEP" contains the values reported by the algorithm in Ref.~\cite{pra86.032307}, and the last column (``Prop1") denotes the corresponding values obtained by Proposition~1. As can be seen, most of the values are improved by Proposition~1, but few are not. We also mark out the values being improved significantly by Proposition~1 with an ``$\uparrow$''.
}
\end{table}
In this section, we make use of Proposition~1 to test different types of entanglement for various multiparticle quantum states.

For the first example, consider the family of $3\times3$ bound entangled states introduced by P. Horodecki \cite{pla232.333},
\begin{equation}\label{eq:rhoPH}
\rho_{\mathrm{PH}}^{a}=\frac1{8a+1}
\left(
\begin{array}{ccccccccc}
 a & 0 & 0 & 0 & a & 0 & 0 & 0 & a \\
 0 & a & 0 & 0 & 0 & 0 & 0 & 0 & 0 \\
 0 & 0 & a & 0 & 0 & 0 & 0 & 0 & 0 \\
 0 & 0 & 0 & a & 0 & 0 & 0 & 0 & 0 \\
 a & 0 & 0 & 0 & a & 0 & 0 & 0 & a \\
 0 & 0 & 0 & 0 & 0 & a & 0 & 0 & 0 \\
 0 & 0 & 0 & 0 & 0 & 0 & \frac{1+a}{2} & 0 & \frac{\sqrt{1-a^2}}{2} \\
 0 & 0 & 0 & 0 & 0 & 0 & 0 & a & 0 \\
 a & 0 & 0 & 0 & a & 0 & \frac{\sqrt{1-a^2}}{2} & 0 & \frac{1+a}{2}
\end{array}
\right)\!.
\end{equation}
These states are not detected by the PPT criterion and are not distillable, but they are nevertheless entangled for any ${0<a<1}$. Consider the mixture of these states with white noise, i.e., ${\rho(q)=q\rho_{\mathrm{PH}}^{a}+(1-q)\openone/9}$, we then ask for the maximal value of $q$ such that $\rho(q)$ remains separable; see the result in Fig.~\ref{fig:sep_PH}. Compared with the result obtained by the algorithm in Ref.~\cite{pra86.032307} (``SEP"), we get a significant improvement. Even though the values by ``SEP" can be improved with an optimized algorithm,\footnote{We sincerely thank H. Kampermann for sharing with us his optimized code in Ref.~\cite{pra86.032307}.} they are still worse than the values by Proposition~1. Moreover, the upper bound for entanglement reported in Ref.~\cite{prl109.200503} (``OBS1") is very close to the values found by Proposition~1.

In Table~\ref{tab:thres}, more examples are presented. Note that this table is extracted from Ref.~\cite{pra86.032307} for comparison. As we can see, most of the threshold values are improved by Proposition~1, but few are not.
Recently in \cite{arXiv1705.01523}, by using machine learning techniques, the threshold for state $\rho_{\mathrm{UPB}}$ is reported to be $0.8649$. However, a huge number of random extreme points within the separable region is needed by the method in Ref.~\cite{arXiv1705.01523}, which is only useful for the particular state tested.

\section{Appendix D: Proof of Theorem~1}\label{sec:thm1}
We first prove Theorem~1 for the DG algorithm, for which the following two statements are needed; see Fig.~\ref{fig:thm1}.

\begin{statement}
There can only be one quantum state $\rho^{*}$, such that ${\cS\!\circ\!\cG}(\rho^{*})=\rho^{*}$ in the DG algorithm.
\end{statement}
\begin{proof}
Consider the case that ${\cG(\rho^{*})\neq\rho^{*}}$ for a fixed point $\rho^{*}$, otherwise it is trivial because we are already at the optimum. Let ${\epsilon=||\cG(\rho^{*})-\rho^{*}}||$, then the ball ${B_{\epsilon}[\cG(\rho^{*})]}$ centered at $\cG(\rho^{*})$ with radius $\epsilon$ contains only one state in $\cC$, namely, $\rho^{*}$.

Assume, instead, there are two fixed points $\rho^{*}$ and $\tilde{\rho}$, with their corresponding balls ${B_{\epsilon}[\cG(\rho^{*})]}$ and ${B_{\delta}[\cG(\tilde{\rho})]}$. Then, for the line $\vec{l}=\tilde{\rho}-\rho^{*}$ connecting $\rho^{*}$ and $\tilde{\rho}$, we have ${\braket{\cG(\rho^{*})-\rho^{*}}{\vec{l}}\le0}$ and ${\braket{\cG(\tilde{\rho})-\tilde{\rho}}{\vec{l}}\ge0}$, otherwise the two balls would contain more states in $\cC$.

Parametrize the line $\vec{l}$ with some parameter $t\in [0,1]$, and look at the function $\cF(\rho)$ on this line with respect to $t$. We then have ${\partial_t\cF(\rho)|_{t=0}\le0}$ and ${\partial_t\cF(\rho)|_{t=1}\ge0}$, which implies that ${\cF(\rho)\equiv\mbox{Const}}$ due to concavity. However, this contradicts the fact that $\cF(\rho)$ is strictly convex. Thus, Statement 1 is true.
\end{proof}

\begin{figure}[]
\includegraphics[width=0.8\columnwidth]{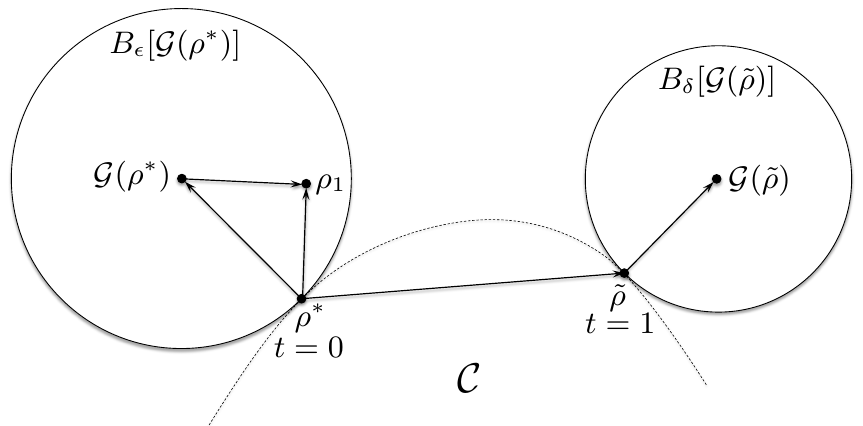}
\caption{\label{fig:thm1}%
For the proof of Theorem 1. See the text in Appendix~D for details.
}
\end{figure}

\begin{statement}
If $\rho^{*}$ maximizes the function $\cF(\rho)$ over $\cC$, then $\rho^{*}$ is a fixed point.
\end{statement}
\begin{proof}
The statement is clear if $\rho^{*}$ is contained inside $\cC$.
So let us assume that the statement is not true and consider
$\rho_1={\cS\!\circ\!\cG}(\rho^{*})$ as the next potential update.
As $\rho^{*}$ lies on the boundary of $\cC$, we have  that all the
states on the line $\vec{l}=\rho_1-\rho^{*}$ belong to $\cC$ because
of convexity.
Moreover, since $\rho_1$ lies within the ball $B_{\epsilon}[\cG(\rho^{*})]$, we have the overlapping ${\braket{\vec{l}}{\cG(\rho^{*})-\rho^{*}}>0}$ which indicates that the function value can still be increased. As a result, there must exist a state $\rho$ on the line $\vec{l}$ such that ${\cF(\rho)>\cF(\rho^{*})}$. This contradicts the assumption that $\rho^{*}$ is the maximum of $\cF(\rho)$. Thus, Statement 2 is true.
\end{proof}

As a consequence of the above two statements, if a fixed point $\rhoML^{\cC}$ is found by the DG algorithm, then this is the solution to the optimization problem in Theorem~1. Thus, Theorem~1 for the DG algorithm is proved. For the APG algorithm, however, the update of our target $\rho$ is based on another state $\sigma$. Thus, Statement~1 has to be modified as the following:
\begin{statement}
There can only be one quantum state $\rho^{*}$, such that ${\cS\!\circ\!\cG[\cS\!\circ\!\cG}(\sigma^{*})]=\rho^{*}$ in the APG algorithm.
\end{statement}
\begin{proof}
By first applying Statement~1, there is only one state $\sigma^{*}$, such that ${\cS\!\circ\!\cG}(\sigma^{*})=\sigma^{*}$. Such a situation in the APG algorithm would trigger the operation `Restart' to reset $\sigma^{*}\equiv\rho^{*}$. Then by using Statement~1 once again, we have only one state $\rho^{*}$, such that ${\cS\!\circ\!\cG[\cS\!\circ\!\cG}(\sigma^{*})]\equiv{\cS\!\circ\!\cG}(\rho^{*})=\rho^{*}$. Thus, Statement~3 is true.
\end{proof}
Therefore, by combing Statements~2 and 3, Theorem~1 for the APG algorithm is also proved.

\section{Appendix E: More simulated examples}\label{sec:examples}
\subsection{Random two-qubit pure state with white noise}
Consider a randomly generated two-qubit pure state $\ket{\phi}$ mixed with white noise,
\begin{equation}
\rho(q)=q\ket{\phi}\bra{\phi}+\frac{1-q}{4}\openone\,.
\end{equation}
In the simulation, we set the noise level ${q=0.9}$, then apply the Pauli scheme. Figure~\ref{fig:2qb} shows the normalized log-likelihood ratios $\lambda/N$ at each iterative step by the two algorithms. As can be seen, the APG algorithm reports better solution than DG does. For comparison, we randomly generated one million two-qubit quantum states, then calculated the minimal $\lambda/N$ value with the help of PPT. We find that this value ($\sim\!0.1305$) is much worse than the values obtained by our algorithms.
\begin{figure}[t]
\includegraphics[width=0.8\columnwidth]{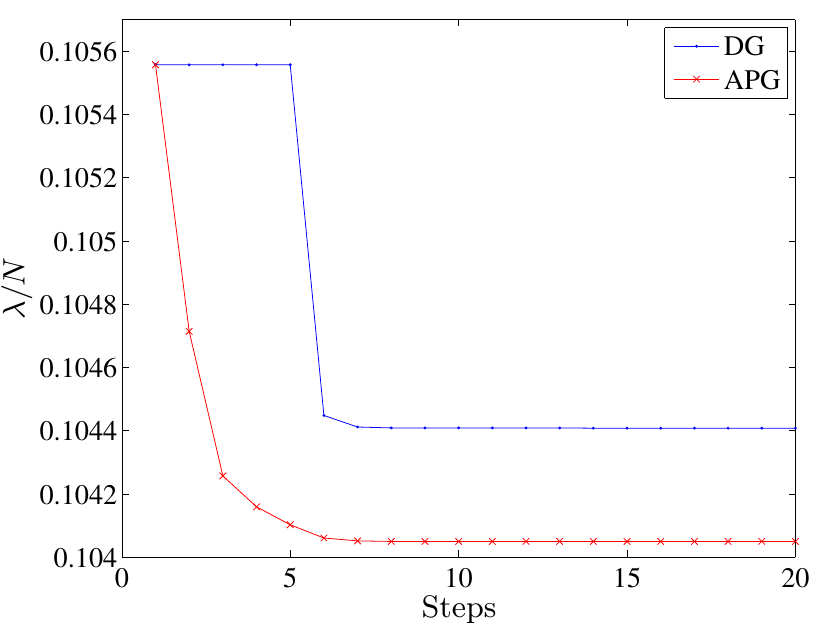}
\caption{\label{fig:2qb}%
The normalized log-likelihood ratios $\lambda/N$ at each iterative step by DG and APG respectively, for the random two-qubit pure state with white noise.
}
\end{figure}
\subsection{Bound entangled state}
In this example, we consider one of the Horodecki ${3\times3}$ bound entangled states $\rho_{\mathrm{PH}}^{a}$ (see Appendix~C and Ref.~\cite{pla232.333}), e.g., $a=0.3$. For each qutrit, we use the symmetric informationally complete (SIC) POVM in $d=3$, thus the overall POVM has nine outcomes. The results are shown in Fig.~\ref{fig:PH}. This example clearly demonstrates that the APG algorithm is capable of resolving the accuracy problem very easily; while the DG algorithm in this case is hard to proceed. 
\begin{figure}[t]
\includegraphics[width=0.8\columnwidth]{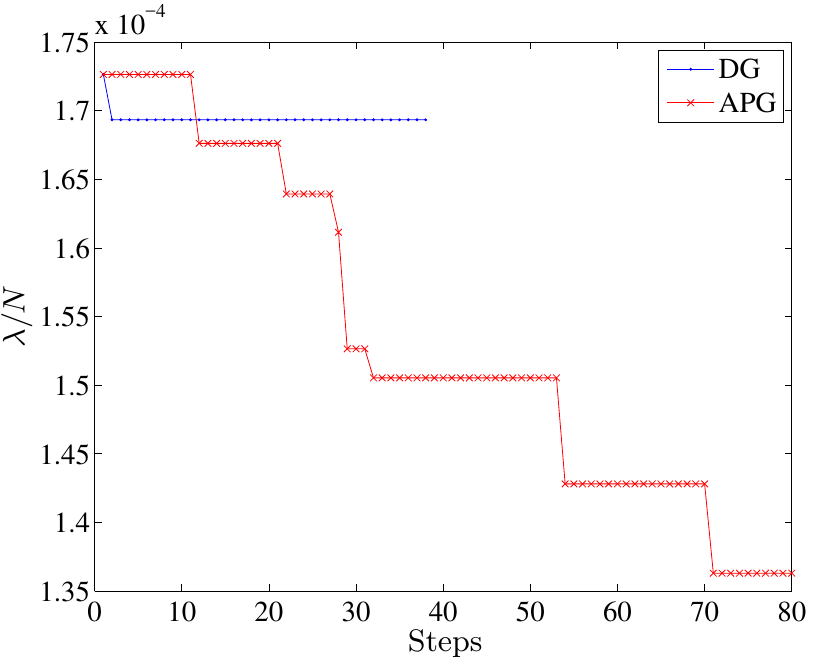}
\caption{\label{fig:PH}%
The normalized log-likelihood ratios $\lambda/N$ at each iterative step by DG and APG respectively, for the Horodecki ${3\times3}$ bound entangled state.
}
\end{figure}
%

\section{Appendix F: Simulated experiments for noisy Smolin states}\label{sec:simu}
Using the noisy Smolin states with various noise levels $q$, we perform simulated tomography experiments. The settings used in the simulation are exactly the same as those in Ref.~\cite{PhotonSmo}, and the number of copies of the true states used is around 4 million.
With the data obtained, we then perform the likelihood-ratio test; see the result shown in Fig.~\ref{fig:simu}.
By doing so, we identify the parameter range ${q\sim[0.35, 0.8]}$, which contains the most-likely candidate states that are expected to show bound entanglement.
For the real experiment in Ref.~\cite{PhotonSmo} that demonstrated bound entanglement, the noise level ${q=0.51}$  is certainly within this range.
As such, our method provides a reliable guidance for the experimentalists to choose the best candidate states for their future experiments.
\begin{figure}[t]
\includegraphics[width=0.8\columnwidth]{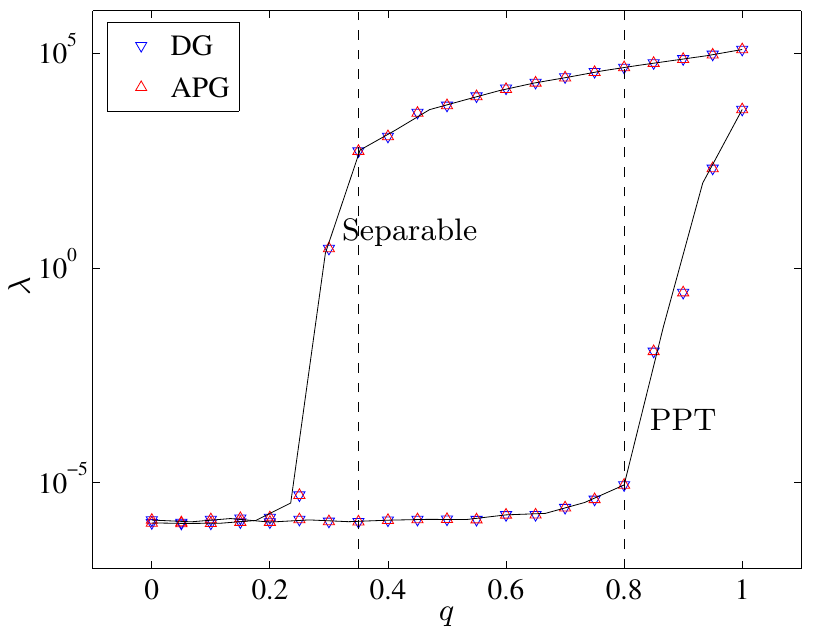}
\caption{\label{fig:simu}%
Likelihood-ratio test for the simulated experiments using the noisy Smolin states $\rho_{\mathrm S}(q)=q\rho_{\mathrm S}+({1-q})\openone/16$.
The parameter range ${q\sim[0.35, 0.8]}$ (dashed vertical lines) encloses the candidate states with very large $\lambda$ values over the separable region but rather small $\lambda$ values over the PPT region, thus the states contained are most likely to be bound entangled.
Note that the fitted black curves passing through the data points are for guidance of the eyes only.
}
\end{figure}
%



\begin{thebibliography}{00}
\bibitem{note1}
Unfortunately, we cannot do hypothesis testing to check for PPT, as the set of negative partial transpose (NPT) states is not convex. A {$p$-value $\approx1$} for the null hypothesis that the state is PPT should not be interpreted as a {$p$-value $\approx0$} for the null hypothesis that the state is NPT.
\end{thebibliography}
\end{document}